\documentclass[showpacs,preprintnumbers,amsmath,amssymb]{revtex4}
\usepackage{amsthm}
\usepackage{enumerate}
\usepackage{epsfig}
 \newtheorem{theorem}{Theorem}[section]
 \newtheorem{lemma}[theorem]{Lemma}
 \newtheorem{corollary}[theorem]{Corollary}
 \newtheorem{definition}[theorem]{Definition}
 \newtheorem{proposition}[theorem]{Proposition}

\newcommand*{\ci}{i}

\newcommand*{\cH}{\mathcal{H}}
\newcommand*{\cI}{\mathcal{I}}

\newcommand*{\cM}{\mathcal{M}}

\newcommand*{\cS}{\mathcal{S}}

\newcommand*{\cY}{\mathcal{Y}}

\newcommand*{\ket}[1]{|#1\rangle}
\newcommand*{\bra}[1]{\langle #1|}
\newcommand*{\proj}[1]{\ket{#1}\bra{#1}}

\newcommand{\braket}[2]{\langle #1|#2\rangle}       

\newcommand{\be}{\begin{equation}}
\newcommand{\ee}{\end{equation}}
\newcommand{\bea}{\begin{eqnarray}}
\newcommand{\eea}{\end{eqnarray}}
\newcommand{\bestar}{\begin{equation*}}
\newcommand{\eestar}{\end{equation*}}
\newcommand{\beastar}{\begin{eqnarray*}}
\newcommand{\eeastar}{\end{eqnarray*}}

\begin{document}

\title{Weak measurement takes a simple form for cumulants}

\author{Graeme \surname{Mitchison}}
\email[]{g.j.mitchison@damtp.cam.ac.uk} \affiliation{Centre for
Quantum Computation, DAMTP,
             University of Cambridge,
             Cambridge CB3 0WA, UK}

\begin{abstract}
A weak measurement on a system is made by coupling a pointer weakly to
the system and then measuring the position of the pointer. If the
initial wavefunction for the pointer is real, the mean displacement of
the pointer is proportional to the so-called weak value of the
observable being measured. This gives an intuitively direct way of
understanding weak measurement. However, if the initial pointer
wavefunction takes complex values, the relationship between pointer
displacement and weak value is not quite so simple, as pointed out
recently by R. Jozsa \cite{J07}. This is even more striking in the
case of sequential weak measurements \cite{MJP07}. These are carried
out by coupling several pointers at different stages of evolution of
the system, and the relationship between the products of the measured
pointer positions and the sequential weak values can become extremely
complicated for an arbitrary initial pointer
wavefunction. Surprisingly, all this complication vanishes when one
calculates the cumulants of pointer positions. These are directly
proportional to the cumulants of sequential weak values. This suggests
that cumulants have a fundamental physical significance for weak
measurement.
\end{abstract}


\maketitle
\pagestyle{plain}

\section{Introduction}

In physics, formal simplicity is often a reliable guide to the
significance of a result. The concept of weak measurement, due to
Aharonov and his coworkers \cite{AharonovRohrlich05,AAV88}, derives
some of its appeal from the formal simplicity of its basic
formulae. One can extend the basic concept to a sequence of weak
measurements carried out at a succession of points during the
evolution of a system \cite{MJP07}, but then the formula relating
pointer positions to weak values turns out to be not quite so simple,
particularly if one allows arbitrary initial conditions for the
measuring system.  I show here that the complications largely
disappear if one takes the cumulants of expected values of pointer
positions; these are related in a formally satisfying way to weak
values, and this form is preserved under all measurement conditions.

The goal of weak measurement is to obtain information about a quantum
system given both an initial state $\ket{\psi_i}$ and a final,
post-selected state $\ket{\psi_f}$. Since weak measurement causes only
a small disturbance to the system, the measurement result can reflect
both the initial and final states. It can therefore give richer
information than a conventional (strong) measurement, including in
particular the results of all possible strong measurements
\cite{OB05,Consortium99}.  To carry out the measurement, a measuring
device is coupled to the system in such a way that the system is only
slightly perturbed; this can be achieved by having a small coupling
constant $g$. After the interaction, the pointer's position $q$ is
measured (or possibly some other pointer observable; e.g. its momentum
$p$). Suppose that, following the standard von Neumann paradigm,
\cite{vonNeumann55}, the interaction between measuring device and
system is taken to be $H_{int}=g\delta(t)pA$, where $p$ is the
momentum of a pointer and the delta function indicates an impulsive
interaction at time $t$. It can be shown \cite{AAV88} that the
expectation of the pointer position, ignoring terms of order $g^2$ or
higher, is
\begin{align}\label{qclassic}
\langle q \rangle=g Re A_w,
\end{align}
where $A_w$ is the {\em weak value} of the observable $A$ given by
\begin{align}
A_w=\frac{\bra{\psi_f}A\ket{\psi_i}}{\braket{\psi_f}{\psi_i}}.
\end{align}

As can be seen, (\ref{qclassic}) has an appealing simplicity, relating
the pointer shift directly to the weak value. However, this formula
only holds under the rather special assumption that the initial
pointer wavefunction $\phi$ is a gaussian, or, more generally, is real
and has zero mean. When $\phi$ is a completely general wavefunction,
i.e. is allowed to take complex values and have any mean value
\cite{J07,MJP07}, equation (\ref{qclassic}) is replaced by
\begin{align}\label{complex-version}
\langle q \rangle=\langle q \rangle_\ci+gRe A_w+gIm A_w\left(\langle pq+qp \rangle_\ci-2 \langle q \rangle_\ci \langle p \rangle_\ci \right),
\end{align}
where, for any pointer variable $x$, $\langle x \rangle_\ci$ denotes
the initial expected value $\bra{\phi}x\ket{\phi}$ of $x$; so for
instance $\langle q \rangle_\ci$ and $\langle p \rangle_\ci$ are the
means of the initial pointer position and momentum,
respectively. (Again, this formula ignores terms of order $g^2$ or
higher.)

Equation (\ref{complex-version}) seems to have lost the simplicity of
(\ref{qclassic}), but we can rewrite it as
\begin{align}\label{firstxi}
\langle q \rangle =\langle q \rangle_\ci +gRe(\xi A_w),
\end{align}
where 
\begin{align} \label{xi1}
\xi=-2i\left(\langle qp \rangle_\ci-\langle q \rangle_\ci \langle p \rangle_\ci \right),
\end{align}
and equation (\ref{firstxi}) is then closer to the form of
(\ref{qclassic}). As will become clear, this is part of a general
pattern.

One can also weakly measure several observables, $A_1, \ldots, A_n$,
in succession \cite{MJP07}. Here one couples pointers at several
locations and times during the evolution of the system, taking the
coupling constant $g_k$ at site $k$ to be small. One then measures
each pointer, and takes the product of the positions $q_k$ of the
pointers. For two observables, and in the special case where the
initial pointer distributions are real and have zero mean, e.g. a
gaussian, one finds \cite{MJP07}
\begin{align} \label{ABmean} 
\langle q_1q_2 \rangle=\frac{g_1g_2}{2}\ Re \left[
(A_2,A_1)_w+(A_1)_w\overline{(A_2)}_w \right],
\end{align}
ignoring terms in higher powers of $g_1$ and $g_2$.
Here $(A_2,A_1)_w$ is the {\em sequential weak value} defined by
\begin{align} \label{ABweakvalue}
(A_2,A_1)_w=\frac{\bra{\psi_f}WA_2VA_1U\ket{\psi_i}}{\bra{\psi_f}WVU\ket{\psi_i}},
\end{align}
where $U$ is a unitary taking the system from the initial state
$\ket{\psi_i}$ to the first weak measurement, $V$ describes the
evolution between the two measurements, and $W$ takes the system to
the final state. (Note the reverse order of operators in $(A_2,A_1)$,
which reflects the order in which they are applied.) If we drop the
assumption about the special initial form of the pointer distribution
and allow an arbitrary $\phi$, then the counterpart of (\ref{ABmean})
becomes extremely complicated: see Appendix, equation \ref{horrible}.

Even the comparatively simple formula (\ref{ABmean}) is not quite
ideal. By analogy with (\ref{qclassic}) we would hope for a formula of
the form $\langle q_1q_2 \rangle \propto Re (A_2,A_1)_w$, but
there is an extra term $(A_1)_w\overline{(A_2)}_w$. What we seek,
therefore, is a relationship that has some of the formal simplicity of
(\ref{qclassic}) and furthermore preserves its form for all
measurement conditions. It turns out that this is possible if we take
the {\em cumulant} of the expectations of pointer positions. As
we shall see in the next section, this is a certain sum of products of
joint expectations of subsets of the $q_i$, which we denote by
$\langle q_1 \ldots  q_n \rangle^c$. For a set of observables, we can define a
formally equivalent expression using sequential weak values, which we
denote by $(A_n, \ldots , A_1)^c_w$. Then the claim is that, up to
order $n$ in the coupling constants $g_k$ (assumed to be all of the
same approximate order of magnitude):
\begin{align} \label{cumulant-equation}
\langle q_1 \ldots  q_n\rangle^c=g_1 \ldots g_n Re \left\{ \xi (A_n, \ldots
, A_1)_w^c\right\},
\end{align}
where $\xi$ is a factor dependent on the initial wavefunctions for
each pointer. Equation (\ref{cumulant-equation}) holds for any
initial pointer wavefunction, though different wavefunctions produce
different values of $\xi$. The remarkable thing is that all the
complexity is packed into this one number, rather than exploding into
a multiplicity of terms, as in (\ref{horrible}).

Note also that (\ref{firstxi}) has essentially the same form as
(\ref{cumulant-equation}) since, in the case $n=1$,
$\langle A \rangle^c_w=A_w$. However, there is an extra term $\langle q
\rangle_\ci$ in (\ref{firstxi}); this arises because the cumulant
for $n=1$ is anomalous in that its terms do not sum to zero.

\section{Cumulants}

Given a collection of random variables, such as the pointer positions
$q_i$, the cumulant $\langle q_1 \ldots q_n \rangle^c$ is a polynomial
in the expectations of subsets of these variables
\cite{KendallStuart77,Royer83}; it has the property that it vanishes
whenever the set of variables $q_i$ can be divided into two
independent subsets. One can say that the cumulant, in a certain
sense, picks out the maximal correlation involving all of the
variables.

We introduce some notation to define the cumulant. Let $x$ be a subset
of the integers $\{1, \ldots, n\}$. We write $\prod_x q$ for
$\prod_{i=1}^{|x|} q_{x(i)}$, where $|x|$ is the size of $x$ and the
indices of the $q$'s in the product run over all the integers $x(i)$
in $x$. Then the cumulant is given by
\begin{align}\label{cumulant}
\langle q_1 \ldots q_n \rangle^c=\sum_{b=\{b_1,\ldots,
b_k\}}a_k\prod_{j=1}^k \left\langle \prod_{b_j}q \right\rangle,
\end{align}
where $b=\{b_1,\ldots, b_k\}$ runs over all partitions of the integers
$\{1, \ldots, n\}$ and the coefficient $a_k$ is given by
\begin{align}\label{coefficients}
a_k=(k-1)!(-1)^{k-1}.
\end{align}

For $n=1$ we have $\langle q \rangle^c=\langle q \rangle$, and for
$n=2$
\begin{align}\label{q2}
\langle q_1 q_2 \rangle^c=\langle q_1q_2 \rangle-\langle q_1 \rangle \langle q_2 \rangle.
\end{align}

There is an inverse operation for the cumulant
\cite{ZZXY06,Royer83}:
\begin{proposition}\label{anti}
\begin{align}\label{anticumulant}
\langle q_1 \ldots q_n \rangle=\sum_{b=\{b_1,\ldots,
b_k\}}\prod_{j=1}^k \left\langle \prod_{b_j}q \right\rangle^c.
\end{align}
\end{proposition}
\begin{proof}
To see that this equation holds, we must show that the term
$\prod_{j=1}^k\langle \prod_{b_j} q \rangle$ obtained by expanding the
right-hand side is zero unless $b$ is the partition consisting of the
single set $\{1, \ldots, n\}$. Replacing each subset $b_j$ by the
integer $j$, this is equivalent to $\sum a_{k_1} \ldots a_{k_r}=0$,
where the sum is over all partitions of $\{1, \ldots, k\}$ by subsets
of sizes $k_1, \ldots, k_r$ and the $a_{k}$'s are given by
(\ref{coefficients}). In this sum we distinguish partitions with
distinct integers; e.g. $\{1,2\},\{3,4\}$ and $\{1,3\},\{2,4\}$. There
are $\binom{k}{k_1 \ldots k_r}(l_1!  \ldots l_k!)^{-1}$ such distinct
partitions with subset sizes $k_1 \ldots k_r$, where $l_i$ is the
number of $k$'s equal to $i$, so our sum may be rewritten as $k! \sum
(-1)^{k_1-1} \ldots (-1)^{k_r-1} (l_1!  \ldots l_k! k_1 \ldots
k_r)^{-1}$, where the sum is now over partitions in the standard sense
\cite{Apostol76}. This is $k!$ times the coefficient of $x^k$ in
\begin{align}
&\left(1+x+\frac{x^2}{2!}+\ldots \right)\left(1+(-x^2/2)+\frac{(-x^2/2)^2}{2!}+\ldots \right)\left(1+(x^3/3)+\frac{(x^3/3)^2}{2!}+\ldots \right)\ldots\\
&=e^{x-x^2/2+x^3/3 \ldots}= e^{log_e(1+x)}=1+x.
\end{align}
Thus the sum is zero except for $k=1$, which corresponds to the
single-set partition $b$.
\end{proof}

\begin{definition}
If $\{1, \ldots, n\}$ can be written as the disjoint union of two
subsets $S_1$ and $S_2$, we say the variables corresponding to these
subsets are independent if 
\begin{align}\label{indep}
\langle \prod_{S_1^\prime}q \prod_{S_2^\prime} q\rangle=\langle \prod_{S_1^\prime} q\rangle \langle \prod_{S_2^\prime} q\rangle,
\end{align}
for any subsets $S_i^\prime \subseteq S_i$. 
\end{definition}
We now prove the characteristic property of cumulants:
\begin{proposition}\label{indep-lemma}
The cumulant vanishes if its arguments can be divided into two
independent subsets.
\end{proposition}
\begin{proof}
For $n=2$ this follows at once from (\ref{q2}) and (\ref{indep}), and
we continue by induction. From (\ref{anticumulant}) and the inductive
assumption for $n-1$, we have
\begin{align}
\langle q_1 \ldots q_n \rangle=\langle q_1 \ldots q_n \rangle^c+\sum_{b=\{b_1,\ldots,
b_k\} \subset S_1}\prod_{j=1}^k \left\langle \prod_{b_j}q \right\rangle^c  \sum_{c=\{c_1,\ldots,
c_l\} \subset S_2}\prod_{j=1}^l \left\langle \prod_{c_j}q \right\rangle^c.
\end{align}
This holds because any term on the right-hand side of
(\ref{anticumulant}) vanishes when any subset of the partition $b$
includes elements of both $S_1$ and $S_2$. Using (\ref{anticumulant})
again, this implies
\begin{align}
\langle q_1 \ldots q_n \rangle=\langle q_1 \ldots q_n
\rangle^c+\langle \prod_{S_1} q\rangle \langle \prod_{S_2} q\rangle,
\end{align}
and by independence, $\langle q_1 \ldots q_n \rangle^c=0$. Thus the
inductive assumption holds for $n$.
\end{proof}
In fact, the coefficients $a_k$ in (\ref{cumulant}) are uniquely
determined to have the form (\ref{coefficients}) by the requirement
that the cumulant vanishes when the variables form two independent
subsets \cite{Percus75,Simon79}.

For $n=2$, the cumulant (\ref{q2}) is just the covariance,
$\langle q_1 q_2 \rangle^c=\langle (q_1-\langle q_1
\rangle)(q_2-\langle q_2 \rangle) \rangle$, and the same is true for
$n=3$, namely $\langle q_1 q_2q_3 \rangle^c=\langle (q_1-\langle q_1
\rangle)(q_2 -\langle q_2 \rangle)(q_3 -\langle q_3 \rangle)
\rangle$. For $n=4$, however, there is a surprise. The covariance is
given by
\begin{align} \label{4covariance}
\langle \prod_{i=1}^4 (q_i-\langle q_i \rangle) \rangle=\langle q_1q_2q_3q_4 \rangle -\sum \langle q_iq_jq_k \rangle\langle q_l \rangle+\sum \langle q_iq_j \rangle\langle q_k \rangle\langle q_l \rangle-3\langle q_1\rangle\langle q_2\rangle\langle q_3 \rangle\langle q_4 \rangle,
\end{align}
where the sums include all distinct combinations of indices, but the
cumulant is
\begin{align} \label{4cumulant}
\langle q_1q_2q_3q_4\rangle^c=\langle q_1q_2q_3q_4 \rangle -\sum \langle q_iq_jq_k \rangle\langle q_l \rangle-\sum \langle q_iq_j \rangle\langle q_kq_l \rangle+2\sum \langle q_iq_j \rangle\langle q_k \rangle\langle q_l \rangle-6\langle q_1\rangle\langle q_2\rangle\langle q_3 \rangle\langle q_4 \rangle,
\end{align}
which includes terms like $\langle q_1q_2 \rangle\langle q_3q_4
\rangle$ that do not occur in the covariance. Note that, if the
subsets $\{1,2\}$ and $\{3,4\}$ are independent, the covariance does
not vanish, since independence implies we can write the first term in
(\ref{4covariance}) as $\langle q_1q_2q_3q_4 \rangle=\langle q_1q_2
\rangle \langle q_3q_4 \rangle$ and there is no cancelling
term. However, as we have seen, the cumulant does contain such a term,
and it is a pleasant exercise to check that the whole cumulant
vanishes.

\section{Sequential weak values and cumulants}

To carry out a sequential weak measurement, one starts a system in an
initial state $\ket{\psi_i}$, then weakly couples pointers at several
times $t_k$ during the evolution of the system, and finally
post-selects the system state $\ket{\psi_f}$. One then measures the
pointers and finally takes the product of the values obtained from
these pointer measurements. It is assumed that one can repeat the
whole process many times to obtain the expectation of the product of
pointer values. If one measures pointer positions $q_k$, for instance,
one can estimate $\langle q_1 \ldots q_n \rangle$, but one could also
measure the momenta of the pointers to estimate $\langle p_1 \ldots
p_n \rangle$.

If the coupling for the $k$th pointer is given by
$H_{int}=\delta(t-t_k)r_kp$, and if the individual initial pointer
wavefunctions are gaussian, or, more generally, are real with zero
mean, then it turns out \cite{MJP07} that these expectations can be
expressed in terms of sequential weak values of order $n$ or less.
Here the sequential weak value of order $n$, $(A_n, \ldots
A_1)_w$, is defined by
\begin{align} \label{general-weakvalue}
(A_n, \ldots A_1)_w=\frac{\bra{\psi_f}U_{n+1}A_nU_n \ldots
A_1U_1\ket{\psi_i}}{\bra{\psi_f}U_{n+1} \ldots U_1\ket{\psi_i}},
\end{align}
where $U_i$ defines the evolution of the system between the
measurements of $A_{i-1}$ and $A_i$.

When the $A_k$ are projectors, $A_k=\proj{x_k}$, we can write the sequential
weak value as \cite{MJP07}
\begin{align} \label{amplitude}
(A_n, \ldots A_1)_w=\frac{\bra{\psi_f}U_{n+1}\ket{x_n}\ \bra{x_n}U_n\ket{x_{n-1}} \ldots
\bra{x_1}U_1\ket{\psi_i}}{\sum_y \bra{\psi_f}U_{n+1}\ket{y_n}\ \bra{y_n}U_n\ket{y_{n-1}} \ldots
\bra{y_1}U_1\ket{\psi_i}}=\frac{\mbox{amplitude}(x)}{\sum_y \mbox{amplitude}(y)},
\end{align}
which shows that, in this case, the weak values has a natural
interpretation as the amplitude for following the path defined by
the $x_k$. Figure \ref{cumulant} shows an example taken from \cite{MJP07}
where the path (labelled by '1' and '2' successively) is a route taken
by a photon through a pair of interferometers, starting by injecting
the photon at the top left (with state $\ket{\psi}_i$) and ending with
post-selection by detection at the bottom right (with final state
$\ket{\psi}_f$).

\begin{figure}[hbtp]
\centerline{\epsfig{file=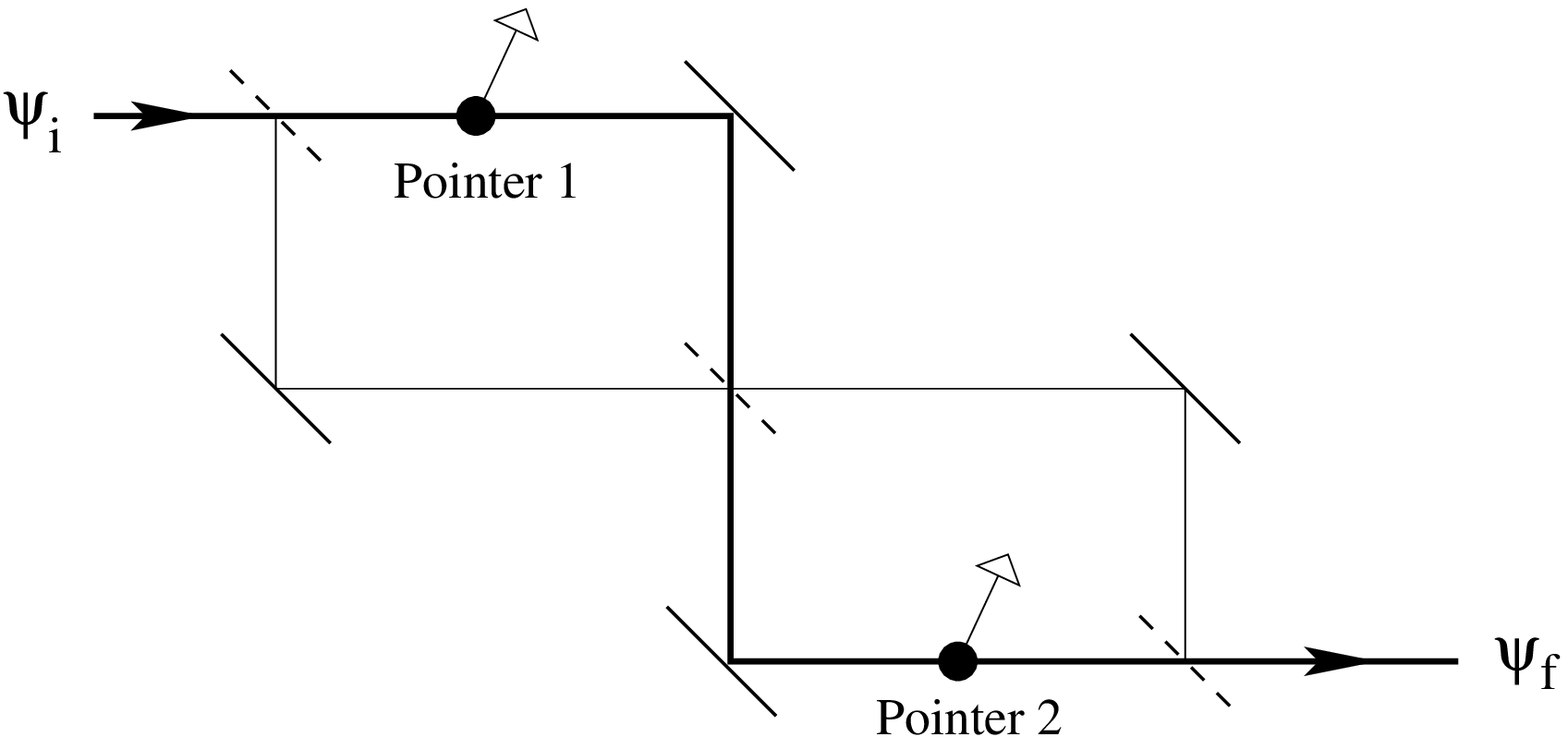,width=0.5\textwidth}}
\caption{\label{cumulantfig}}
\end{figure}

In the last section, the cumulant was defined for expectations of
products of variables. One can define the cumulant for other entities
by formal analogy; for instance for density matrices \cite{ZZXY06}, or
hypergraphs \cite{Royer83}. We can do the same for sequential weak
values, defining the cumulant by (\ref{cumulant}) with $\langle
\prod_{b_j}q\rangle$ replaced by $(\overleftarrow{A_{b_j(|b_j|)},
\ldots A_{b_j(1)}})_w$, where the arrow indicates that the indices,
which run over the subset $b_j$, are arranged in ascending order from
right to left. For example, for $n=1$, $(A_w)^c=A_w$, and for $n=4$
\begin{align}\label{4weak}
(A_4,A_3,A_2,A_1)^c_w&=(A_4,A_3,A_2,A_1)_w-\sum (\overleftarrow{A_i,A_j,A_k})_w(A_l)_w-\sum (\overleftarrow{A_i,A_j})_w(\overleftarrow{A_k,A_l})_w\\
\nonumber &+2\sum (\overleftarrow{A_i,A_j})_w(A_k)_w(A_l)_w-6(A_1)_w(A_2)_w(A_3)_w(A_4)_w.
\end{align}
There is a notion of independence that parallels (\ref{indep}): given
a disjoint partition $S_1 \cup S_2=\{1, \ldots, n\}$ such that
\begin{align}
(\overleftarrow{A_{S_1^\prime \cup S_2^\prime}})_w=(\overleftarrow{A_{S_1^\prime}})_w(\overleftarrow{A_{S_2^\prime}})_w,
\end{align}
for any subsets $S_i^\prime \subseteq S_i$, then we say the
observables labelled by the two subsets are {\em weakly independent}. There is then an analogue of Lemma \ref{indep-lemma}:
\begin{lemma}
The cumulant $(A_n, \ldots, A_1)_w^c$ vanishes if the $A_k$ are weakly
independent for some subsets $S_1$, $S_2$.
\end{lemma}
As an example of this, if one is given a bipartite system $\cH^A
\otimes \cH^B$, and initial and final states that factorise as
$\ket{\psi_i}=\ket{\psi_i}^A \otimes \ket{\psi_i}^B$ and
$\ket{\psi_f}=\ket{\psi_f}^A \otimes \ket{\psi_f}^B$, then observables
on the $A$- and $B$-parts of the system are clearly weakly
independent.
Another class of examples comes from what one might describe as a
``bottleneck'' construction, where, at some point the evolution of the
system is divided into two parts by a one-dimensional projector (the
bottleneck) and its complement, and the post-selection excludes the
complementary part. Then, if all the measurements before the projector
belong to $S_1$ and all those after the projector belong to $S_2$, the
two sets are weakly independent. This follows because we can write
\begin{align*}
(\overleftarrow{A_{S_1^\prime \cup S_2^\prime}})_w 
&=\frac{\bra{\psi_f}U_{n+1}A_n \ldots U_{k+1}A_kW_k\proj{\psi_b}V_kA_{k-1}  \ldots A_1U_1\ket{\psi_i}}
{\bra{\psi_f}U_{n+1} \ldots U_{k+1}W_k\proj{\psi_b}V_k  \ldots U_1\ket{\psi_i}} \\
&=\frac{\bra{\psi_f}U_{n+1}A_n \ldots U_{k+1}A_kW_k\proj{\psi_b}V_k  \ldots U_1\ket{\psi_i}}
{\bra{\psi_f}U_{n+1} \ldots U_{k+1}W_k\proj{\psi_b}V_k  \ldots U_1\ket{\psi_i}} \ \ \frac{\bra{\psi_f}U_{n+1} \ldots U_{k+1}W_k\proj{\psi_b}V_kA_{k-1} \ldots A_1U_1\ket{\psi_i}}
{\bra{\psi_f}U_{n+1} \ldots U_{k+1}W_k\proj{\psi_b}V_k  \ldots U_1\ket{\psi_i}}\\
&=(\overleftarrow{A_{S_1^\prime}})_w(\overleftarrow{A_{S_2^\prime}})_w,
\end{align*}
where $W_k\proj{\psi_b}V_k$ is the part of $U_k$ lying in the
post-selected subspace. As an illustration of this, suppose we add a
connecting link (Figure \ref{bottleneckfig}, ``$L$'') between the two
interferometers in Figure \ref{cumulantfig}, so $\proj{\psi_b}$, the
bottleneck, is the projection onto $L$, and post-selection discards
the part of the wavefunction corresponding to the path
$L^\prime$. Then measurements at `1' and `2' are weakly independent;
in fact $(A_1)_w=1/2$, $(A_2)_w=1/2$ and $(A_2,A_1)_w=1/4$. Note that
the same measurements are {\em not} independent in the double
interferometer of Figure \ref{cumulantfig}, where $(A_1)_w=0$,
$(A_2)_w=0$, and yet, surprisingly, $(A_2,A_1)_w=-1/2$, \cite{MJP07}.

\begin{figure}[hbtp]
\centerline{\epsfig{file=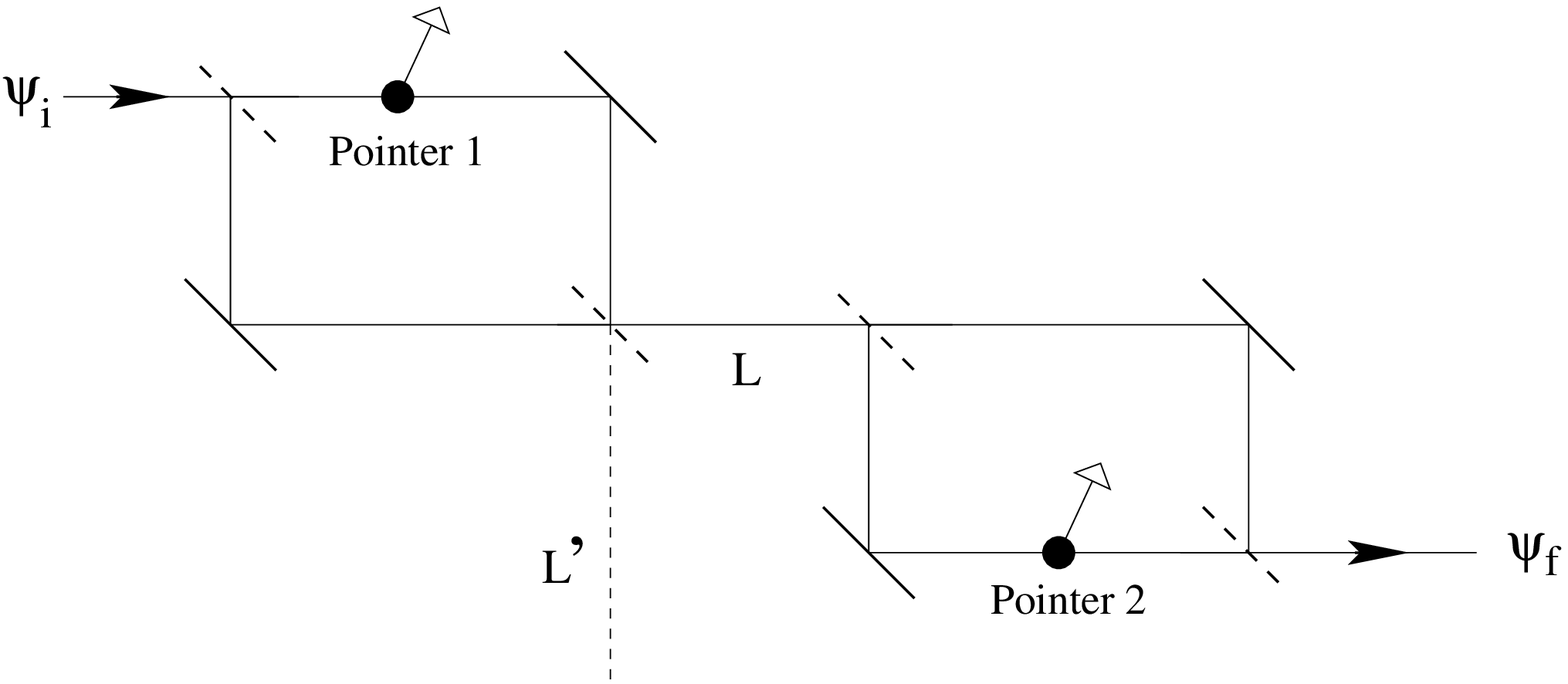,width=0.6\textwidth}}
\caption{\label{bottleneckfig}}
\end{figure}

\section{The main theorem}\label{theorem-section}
Consider $n$ system observables $A_1, \ldots, A_n$. Suppose $s_k$, for
$k=1, \dots, n$, are observables of the $k$th pointer, namely
Hermitian functions $s_k(q_k,p_k)$ of pointer position $q_k$ and
momentum $p_k$, and the interaction Hamiltonian for the weak
measurement of system observable $A_k$ is $H_k=g_ks_kA_k$, where $g_k$
is a small coupling constant (all $g_k$ being assumed of the same
order of magnitude $g$). Suppose further that the pointer observables
$r_k$ are measured after the coupling. Let $\phi_k$ be the $k$-th
pointer's initial wave-function. For any variable $x_k$ associated to
the $k$-th pointer, write $\langle x_k \rangle_\ci$ for
$\bra{\phi_k}x_k \ket{\phi_k}$.

We are now almost ready to state the main theorem, but first need to
clarify the measurement procedure. When we evaluate expectations of
products of the $r_k$ for different sets of pointers, for instance
when we evaluate $\langle r_1 r_2\rangle$, we have a choice. We could
either couple the entire set of $n$ pointers and then select the data
for pointers 1 and 2 to get $\langle r_1 r_2\rangle$. Or we could
carry out an experiment in which we couple just pointers 1 and 2 to
give $\langle r_1 r_2\rangle$. These procedures give different
answers. For instance, if we couple three pointers and measure
pointers 1 and 2 to get $\langle r_1 r_2\rangle$, in addition to the
terms in $g_1$, $g_2$ and $g_1g_2$ we also get terms in $g_2g_3$ and
$g_1g_3$ involving the observable $A_3$. This means we get a different
cumulant $\langle r_1 \ldots r_n \rangle^c$, depending on the
procedure used. In what follows, we regard each expectation as being
evaluated in a separate experiment, with only the relevant pointers
coupled. It will be shown elsewhere that, with the alternative
definition, the theorem still holds but with a different value of the
constant $\xi$.

\begin{theorem}[Cumulant theorem]\label{main-theorem}
For $n \ge 2$, for any pointer observables $r_k$ and $s_k$, and for any
initial pointer wavefunctions $\phi_k$, up to total order $n$ in the $g_k$,
\begin{align}\label{main-result}
\langle r_1 \ldots r_n \rangle^c= g_1 \ldots g_n Re \left\{ \xi
(A_n, \ldots, A_1)^c_w\right\},
\end{align}
where $\xi$ (sometimes written more explicitly as $\xi_{r_1 \ldots
r_n}$) is given by
\begin{align}
\label{xi} \xi=2(-i)^n \left( \prod_{k=1}^n \langle r_ks_k \rangle_\ci-\prod_{k=1}^n \langle r_k \rangle_\ci \langle s_k \rangle_\ci \right).
\end{align}
For $n=1$ the same result holds, but with the extra term $\langle r
\rangle_i$:
\begin{align} \label{main-result1}
\langle r \rangle =\langle r \rangle_i +g Re (\xi A_w).
\end{align}
\end{theorem}
\begin{proof}
We use the methods of \cite{MJP07} to calculate the expectations of
products of pointer variables for sequential weak measurements. Let
the initial and final states of the system be $\ket{\psi_i}$ and
$\ket{\psi_f}$, respectively. Consider some subset $b=\{b_1, \ldots,
b_\kappa\}$ of $\{1, \ldots, n\}$, with $b_1 \le b_2 \le \ldots \le
b_\kappa$. The state of the system and the pointers $b_1, \ldots,
b_\kappa$ after the coupling of those pointers is
\begin{align}\label{totalstate}
\Psi_{\cS,\cM}=U_{n+1}\dots U_{b_\kappa+1}
e^{-ig_{b_\kappa}s_{b_\kappa}A_{b_\kappa}}U_{b_\kappa} \ldots
e^{-ig_{b_1}s_{b_1}A_{b_1}}U_{b_1}\dots U_1 \ket{\psi_i}\phi_{b_1}(r_{b_1}) \ldots \phi_{b_\kappa}(r_{b_\kappa}),
\end{align}
and following post-selection by the system state $\ket{\psi_f}$, the
state of the pointers is
\begin{align} \label{bigstate}
\Psi_\cM=\bra{\psi_f}U_{n+1}\dots U_{b_\kappa+1}
e^{-ig_{b_\kappa}s_{b_\kappa}A_{b_\kappa}}U_{b_\kappa} \ldots
e^{-ig_{b_1}s_{b_1}A_{b_1}}U_{b_1}\dots U_1\ket{\psi_i}\phi_{b_1}(r_{b_1}) \ldots \phi_{b_\kappa}(r_{b_\kappa}).
\end{align}
Expanding each exponential, we have
\begin{align} \label{expectation}
\langle r_{b_1} \ldots r_{b_\kappa} \rangle &=\frac{\int
\overline{\Psi}_\cM r_{b_1} \ldots r_{b_\kappa} \Psi_\cM dr_{b_1}
\ldots dr_{b_\kappa}}{\int |\Psi_\cM|^2 dr_{b_1} \ldots
dr_{b_\kappa}},\\
\label{sumratio}&=\frac{\sum_{i_1 \ldots i_n \in b\ ;\ j_1 \ldots j_n \in b} \ \alpha_{i_1 \ldots i_n}\overline{\alpha}_{j_1 \ldots j_n}u^{b_1}_{i_{b_1}j_{b_1}} \ldots u^{b_\kappa}_{i_{b_\kappa}j_{b_\kappa}}}{\sum_{i_1 \ldots i_n \in b\ ;\ j_1 \ldots j_n \in b} \ \alpha_{i_1 \ldots i_n}\overline{\alpha}_{j_1 \ldots j_n}v^{b_1}_{i_{b_1}j_{b_1}} \ldots v^{b_\kappa}_{i_{b_\kappa}j_{b_\kappa}}},
\end{align}
where $i_k \ge 0$ are integers, $i_1, \ldots, i_n \in b$ means that
$i_l=0$ for $l \notin b$, and
\begin{align}
\label{alpha} \alpha_{i_1 \ldots i_n}&=\left(\prod_{k=1}^n g_k^{i_k}\right)\ (A^{i_n}_n, \ldots A^{i_1}_1)_w, \\ 
\label{u} u^k_{l m}&=\int (m!)^{-1}\overline{(-is_k)^{m}\phi_k(r_k)}r_k(l!)^{-1}(-is_k)^{l}\phi_k(r_k)dr_k,\\
\label{v} v^k_{l m}&=\int (m!)^{-1}\overline{(-is_k)^{m}\phi_k(r_k)}(l!)^{-1}(-is_k)^{l}\phi_k(r_k)dr_k.
\end{align}
Let us write (\ref{sumratio}) as
\begin{align}
\langle r_{b_1} \ldots r_{b_\kappa} \rangle =\frac{\sum_{{\bf i} \in b, {\bf j} \in b} x_{{\bf i};{\bf j}}}{\sum_{{\bf i} \in b, {\bf j} \in b} y_{{\bf i};{\bf j}}},
\end{align}
where 
\begin{align}
\label{x} x_{{\bf i};{\bf j}}&=\alpha_{i_1 \ldots i_n}\overline{\alpha}_{j_1 \ldots j_n}u^{b_1}_{i_{b_1}j_{b_1}} \ldots u^{b_\kappa}_{i_{b_\kappa}j_{b_\kappa}},\\ 
\label{y} y_{{\bf k};{\bf
l}}&=\alpha_{k_1 \ldots k_n}\overline{\alpha}_{l_1 \ldots
l_n}v^{b_1}_{i_{b_1}j_{b_1}} \ldots
v^{b_\kappa}_{i_{b_\kappa}j_{b_\kappa}},
\end{align}
and ${\bf i}$ denotes the index set $\{i_1 \ldots i_n\}$,
etc.. Define
\begin{align}
X_b=\sum_{{\bf i} \in b, {\bf j} \in b} x_{{\bf i};{\bf j}},\ \ Y_b=\sum_{{\bf i} \in b, {\bf j} \in b} y_{{\bf i};{\bf j}}.
\end{align}
Then 
\begin{align}\label{CXY1}
\langle r_1, \ldots , r_n \rangle^c&=\sum_{b_1, \ldots , b_k} (k-1)!(-1)^{k-1} \prod_{l=1}^k \langle r_{b_l(1)} \ldots r_{b_l(|b_l|)} \rangle \\ \label{CXY2}
&=\sum_{b_1, \ldots , b_k} (k-1)!(-1)^{k-1}
\prod_{l=1}^k \frac{X_{b_l}}{Y_{b_l}}.
\end{align}
Set $\cY=\prod_{b \subset \{1,\ldots , n\}} Y_b$, where $b$ in the
product ranges over all distinct subsets of the integers $\{1,\ldots ,
n\}$. Then $\cY \langle r_1 \ldots r_n \rangle^c$ is an (infinite) weighted sum of
terms 
\begin{align} \label{z}
z_\cI=(x_{{\bf i}(1);{\bf j}(1)} \ldots x_{{\bf i}(m);{\bf
j}(m)})(y_{{\bf k}(1);{\bf l}(1)} \ldots y_{{\bf k}(m^\prime);{\bf
l}(m^\prime)}),
\end{align}
where 
\begin{align}\label{index}
\cI &=\cI_i\cup \cI_j \cup \cI_k \cup \cI_l\\ &=\{{\bf i}(1), \ldots,
{\bf i}(m)\} \cup \{{\bf j}(1), \ldots, {\bf j}(m)\} \cup \{{\bf
k}(1), \ldots, {\bf k}(m^\prime)\} \cup \{{\bf l}(1), \ldots, {\bf
l}(m^\prime)\} \nonumber
\end{align}
denotes the set of all the index sets that occur in $z_\cI$. The
strategy is to show that, when the size of the index set $\cI$ is less
than $n$, the coefficient of $z_\cI$ vanishes; by (\ref{alpha}) this
implies that all coefficients of order less than $n$ in $g$ vanish. We
then look at the index sets of size $n$, corresponding to terms of
order $g^n$, and show that the relevant terms sum up to the right-hand
side of (\ref{main-result}). But if $\cY \langle r_1 \ldots r_n
\rangle^c=g^nx+O(g^{n+1})$ for some x, then we also have $\langle r_1
\ldots r_n \rangle^c=g^nx+O(g^{n+1})$, since $\cY=1+O(g)$.

Let $b=\{b_1, \ldots, b_s\}$ be a partition of $\{1, \ldots, n\}$. We
say that $b$ is a {\em valid} partition for $\cI$ if
\begin{enumerate}[(i)]
\item For each $r$ with $1 \le r \le m$, ${\bf i}(r) + {\bf j}(r) \in
b_l$, for some $b_l$, and we can associate a distinct $b_l$ to each
$r$. (Here ${\bf i} + {\bf j}$ means the index set
$\{i_1+j_1, \ldots i_n+j_n\}$.)
\item For each $r$ with $1 \le r \le m^\prime$, ${\bf k}(r) + {\bf
l}(r) \in S$, for some subset $S \subset \{1,\ldots , n\}$ that is
not in the partition $b$, i.e. for which $S \ne b_l$ for any $l$, and
we can associate a distinct $S$ to each $r$. Let $\gamma(\cI,b)$ be
the number of ways of associating a subset $S$ to each $r$.
\end{enumerate}

\begin{lemma} \label{vanishing}
The coefficient of $z_\cI$ in $\cY \langle r_1 \ldots r_n \rangle^c$ is zero if all
the index sets in $\cI$ have a zero at some position $r$.
\end{lemma}
\begin{proof}
If we expand $\cY \langle r_1 \ldots r_n \rangle^c$ using (\ref{CXY2}), each term
in this expansion is associated with a partition $b$ of $\{1, \ldots,
n\}$. Let $b$ be a valid partition for $\cI$, and let $c=\{c_1,
\ldots, c_s\}$ denote the partition derived from $b$ by removing $r$
from the subset $b_l$ that contains it, and deleting that subset if it
contains only $r$. Then the following partitions include $b$ and are
all valid :
\begin{align}\label{cancel}
&c^{(1)}=\{(rc_1),c_2, \ldots ,c_s\}\\\nonumber
&c^{(2)}=\{c_1,(rc_2), \ldots ,c_s\}\\\nonumber
&\ldots\ldots\ldots \\\nonumber
&c^{(s)}=\{c_1,c_2, \ldots ,(rc_s)\}\\\nonumber
&\\\nonumber
&c^{(s+1)}=\{r,c_1,c_2, \ldots ,c_s\}.
\end{align}
Each partition $c^{(i)}$, for $1 \le i \le s+1$ contributes
$\gamma(\cI,b)$ to the coefficient of $z_\cI$ in $\cY \prod_{l=1}^k
X_{c^{(i)}}/Y_{c^{(i)}}$, and since this term has coefficient
$(s-1)!(-1)^{(s-1)}$ in (\ref{CXY2}) for partitions $c^{(1)}, c^{(2)},
\ldots c^{(s)}$, and $s!(-1)^s$ for $c^{(s+1)}$, the sum of all
contributions is zero.
\end{proof}

From equations (\ref{alpha}) and (\ref{index}), the power of $g$ in
the term $z_\cI$ is $|I|=|I_i|+|I_j|+|I_k|+|I_l|$. This, together with
the preceding Lemma, implies that the lowest order non-vanishing terms
in $\cY \langle r_1 \ldots r_n \rangle^c$ are $z_\cI$'s that have a
'1' occurring once and once only in each position; we call these {\em
complete lowest-degree} terms.
\begin{lemma} \label{one-index-set}
The coefficient of a complete lowest-degree term $z_\cI$ in $\cY
\langle r_1 \ldots r_n \rangle^c$ is zero unless only one of the four classes of
indices in $\cI$, viz. $\cI_i$, $\cI_j$, $\cI_k$ or $\cI_l$, has
non-zero terms.
\end{lemma}
\begin{proof}
Consider first the case where the indices in $\cI_j$ and $\cI_l$ are
zero, and where both $\cI_i$ and $\cI_k$ have some non-zero
indices. Let $b=\{b_1, \ldots, b_r\}$ be the partition whose subsets
consists of the non-zero positions in index sets ${\bf i}(t)$ in
$\cI_i$, and let $c=\{c_1, \ldots, c_s\}$ be some partition of the
remaining integers in $\{1, \ldots, n\}$. Suppose $s \le r$. Then we
can construct a set of partitions by mixing $b$ and $c$; these have
the form
\begin{align} \label{mix}
d^{(w)}=\{c_{i_1}, \dots, c_{i_t},(x_1b_1), \ldots, (x_rb_r)\},
\end{align}
where each $x_i$ is either empty or consists of some $c_i$, and all
the subsets $c_i$ are present once only in the partition. If any
$d^{(w)}$ is eligible, all the other mixtures will also be
eligible. Furthermore, the set of all eligible partitions can be
decomposed into non-overlapping subsets of mixtures obtained in this
way.

Any mixture $d^{(w)}$ gives the same value of $\gamma(\cI,d^{(w)})$,
which we denote simply by $\gamma$; so to show that all the
contributions to the coefficient of $z_\cI$ cancel, we have only to
sum over all the mixtures, weighting a partition with $t$ subsets by
$(t-1)!(-1)^{t-1}$. This gives
\begin{align*}
\mbox{Coefficient of }z_\cI&=\gamma \sum_{i=0}^s (s+r-1)!(-1)^{s+r-i}\binom{s}{i}\binom{r}{i}i!\\
&=\gamma (-1)^{s+r-1}s! \sum (s+r-i-1) \ldots (s-i+1) \binom{r}{i}(-1)^i,\\
&=\gamma (-1)^{s+r-1}s! \frac{\partial^{r-1}}{\partial x^{r-1}} \left\{x^{s-1}(x-1)^r \right\} \vert_{x=1}=0.
\end{align*}

The above argument applies equally well to the situation where $\cI_i$
and $\cI_l$ both have some non-zero indices and indices in $\cI_j$ and
$\cI_k$ are zero. If the non-zero indices are present in $\cI_i$ and
$\cI_j$, we can take any eligible partition $a=\{a_1, \ldots, a_r\}$
and divide each subset $a_k$ into two subsets $b_k$ and $c_k$ with the
indices from $\cI_i$ in $b_k$ and those from $\cI_j$ in $c_k$. All
the mixtures of type (\ref{mix}) are eligible, and they include the
original partition $a$. By the above argument, the coefficients of
$z_\cI$ arising from them sum to zero. Other combinations of indices
are dealt with similarly.

Note that, for $n=4$ and for the index sets $(1,1,0,0) \in \cI_i$ and
$(0,0,1,1) \in \cI_j$, the ``mixture'' argument shows that coefficient
of $z_\cI$ coming from $\langle r_1r_2r_3r_4\rangle$ cancels that
coming from $\langle r_1r_2 \rangle \langle r_3r_4\rangle$ to give
zero. This cancellation occurs with the cumulant (\ref{4cumulant}),
but not with the covariance (\ref{4covariance}), where the term
$\langle r_1r_2 \rangle \langle r_3r_4\rangle$ is absent.

\end{proof}

The only terms that need to be considered, therefore, are complete
lowest-degree terms with non-zero indices only in one of the sets
$\cI_i$, $\cI_j$, $\cI_k$ and $\cI_l$. It is easy to calculate the
coefficients one gets for such terms. Consider the case of $\cI_i$. We
only need to consider the single partition $b$ whose subsets are the
index sets of $\cI_i$. For this partition, by (\ref{z}), (\ref{x}) and
(\ref{y}),
\begin{align}
z_\cI=\prod_{e=1}^t \alpha_{{\bf i}(e)}\prod_{k=1}^n
u^k_{1,0}v^k_{0,0}=g_1 \ldots g_n\prod_{e=1}^t \left(A_{{\bf i}(e)(|{\bf i}(e)|)}, \ldots ,A_{{\bf i}(e)(1)}\right)_w\prod_{k=1}^n \langle r_ks_k \rangle_\ci
\end{align}
From (\ref{CXY2}), $z_\cI$ appears in $\cY \langle r_1 \ldots r_n
\rangle^c$ with a coefficient $(t-1)!(-1)^{t-1}$. So, summing over all
$z_\cI$ with indices in $\cI_i$, one obtains $g_1 \ldots g_n(A_n,
\ldots, A_1)^c_w\prod_{k=1}^n (-i\langle r_ks_k
\rangle_\ci)$. Similarly, from (\ref{alpha}), (\ref{u}) and (\ref{v}),
summing over the $z_\cI$ with indices in $\cI_j$ gives the complex
conjugate of $g_1 \ldots g_n (A_n, \ldots, A_1)^c_w\prod_{k=1}^n
(-i\langle r_ks_k \rangle_\ci)$. Thus $\cI_i$ and $\cI_j$ together
give $g_1 \ldots g_n (2\prod_{k=1}^n (-i\langle r_ks_k \rangle_\ci))
Re \left\{(A_n, \ldots, A_1)^c_w \right\}$.

This corresponds to (\ref{main-result}), but with only the first half
of $\xi$ as defined by (\ref{xi}). The rest of $\xi$ comes from the
index sets $\cI_k$ and $\cI_l$. However, the sum of the coefficients
of $z_\cI$ for the same index set in $\cI_i$ and $\cI_k$ is zero. This
is true because, for any complete lowest degree index set, the sum of
coefficients for all $z_\cI$ with the indices divided in any manner
between $\cI_i$ and $\cI_k$ is zero, being the number ways of
obtaining that index set from $\cY$ times $\sum_{t=1}^n
(t-1)(-1)^{t-1}$. But by Lemma \ref{one-index-set}, the coefficient of
$z_\cI$ is zero unless the index set comes wholly from $\cI_i$ or
$\cI_k$. Now (\ref{z}), (\ref{x}) and (\ref{y}) tell us that, for an
index set in $\cI_k$,
\begin{align}
z_\cI=\prod_{e=1}^t \alpha_{{\bf i}(e)}\prod_{k=1}^n
u^k_{0,0}v^k_{1,0}=g_1 \ldots g_n\prod_{e=1}^t \left(A_{{\bf i}(e)(|{\bf i}(e)|)}, \ldots ,A_{{\bf i}(e)(1)}\right)_w\prod_{k=1}^n \langle r_k \rangle_\ci\langle s_k \rangle_\ci,
\end{align}
and from the above argument, this appears appears in $\cY \langle r_1
\ldots r_n \rangle^c$ with coefficient $-(t-1)!(-1)^{t-1}$. Again, the
index sets in $\cI_l$ give the complex conjugate of those in
$\cI_k$. Thus we obtain the remaining half of $\xi$, which proves
(\ref{main-result}) for $n \ge 2$. For $n=1$ the constant terms (of
order zero in $g$) in $\cY \langle r \rangle$ do not vanish, but the
proof goes through if we consider $\cY(\langle r \rangle-\langle r
\rangle_\ci)$ instead.

\end{proof}

\section{Exploring the theorem}

Consider first the simplest case, where $n=1$ and $r=q$. We take
$H_{int}=g \delta(t) pA$ throughout this section, so $s=p$. Then
(\ref{main-result1}) and (\ref{xi}) give
\begin{align}\label{qmean}
\langle q \rangle =\langle q \rangle_\ci +gRe(\xi_q A_w)\qquad  \mbox{with} \qquad
\xi_q=-2i\left(\langle qp \rangle_\ci-\langle q \rangle_\ci \langle p \rangle_\ci \right),
\end{align}
which we have already seen as equations (\ref{firstxi}) and
(\ref{xi1}).  If we measure the pointer momentum, so $r=p$, we find
\begin{align}\label{pmean}
\langle p \rangle =\langle p\rangle_\ci+gRe(\xi_p A_w) \qquad \mbox{with} \qquad
\xi_p=-2i(\langle p^2\rangle_\ci-\langle p\rangle_\ci^2),
\end{align}
which is equivalent to the result obtained in \cite{J07}.

For two variables, our theorem for $r_1=q_1, r_2=q_2$, is
\begin{align}\label{qq}
\langle q_1q_2 \rangle^c=g_1g_2 Re (\xi_{qq} (A_2,A_1)^c_w),
\end{align}
with
\begin{align}\label{xiqq}
\xi_{qq}=2(\langle q_1\rangle_\ci\langle p_1\rangle_\ci\langle q_2\rangle_\ci\langle p_2\rangle_\ci-\langle q_1p_1\rangle_\ci\langle q_2p_2\rangle_\ci).
\end{align}
The calculations in the Appendix allow one to check (\ref{qq}) and
(\ref{xiqq}) by explicit evaluation; see (\ref{explicit}). Note in
passing that, if one writes $\Delta q=\sqrt{\langle (q_1- \langle q_1
\rangle)^2\rangle}$, the Cauchy-Schwarz inequality
\begin{align*}
\{\langle q_1 q_2 \rangle^c \}^2=\left\{\langle (q_1-\langle q_1 \rangle)(q_2-\langle q_2 \rangle)\right\}^2 \le \langle (q_1- \langle q_1 \rangle)^2\rangle \langle (q_2- \langle q_2 \rangle)^2\rangle
\end{align*}
implies a Heisenberg-type inequality
\begin{align*}
\Delta q_1 \Delta q_2 \ge g_1g_2 Re \{ \xi_{qq} (A_2,A_1)^c_w\},
\end{align*}
relating the pointer noise distributions of two weak measurements
carried out at different times during the evolution of the system.

When one or both of the $q_k$ in (\ref{qq}) is replaced by the
pointer momentum $p_k$, we get
\begin{align}
\label{qp} \langle q_1 p_2 \rangle^c&=g_1g_2 Re \left( \xi_{qp} (A_2,A_1)^c_w\right),\\
\label{pp} \langle p_1 p_2 \rangle^c&=g_1g_2 Re \left( \xi_{pp} (A_2,A_1)^c_w\right),
\end{align}
with \begin{align}
\label{xiqp} \xi_{qp}&=-2\left(\langle q_1p_1\rangle_\ci\langle
p_2^2\rangle_\ci-\langle q_1\rangle_\ci\langle p_1\rangle_\ci
\langle p_2\rangle_\ci^2 \right),\\
\label{xipp} \xi_{pp}&=-2\left(\langle
p_1^2\rangle_\ci\langle p_2^2\rangle_\ci-\langle
p_1\rangle_\ci^2 \langle p_2\rangle_\ci^2 \right).
\end{align}

Consider now the special case where $\phi$ is real with zero
mean. Then the very complicated expression for $\langle q_1q_2
\rangle$ in (\ref{horrible}) reduces to
\begin{align} 
\label{2q} \langle q_1q_2 \rangle=\frac{g_1g_2}{2}\ Re \left[
(A_2,A_1)_w+(A_1)_w({\bar A}_2)_w \right],
\end{align}
as shown in \cite{MJP07}. Two further examples from \cite{MJP07} are
\begin{align} 
\label{3q} \langle q_1q_2q_3 \rangle&=\frac{g_1g_2g_3}{4}\ Re 
\left[(A_3,A_2,A_1)_w+(A_3,A_2)_w({\bar A}_1)_w+(A_3,A_1)_w(({\bar
A}_2)_w+(A_2,A_1)_w({\bar A}_3)_w \right],\\
\label{4q} \langle q_1q_2q_3q_4 \rangle&=\frac{g_1g_2g_3g_4}{8}\ Re \left[
(A_4,A_3,A_2,A_1)_w+(A_4,A_3,A_2)_w({\bar A}_1)_w+\ldots +(A_4,A_3)_w(\overline{A_2,A_1})_w+\ldots \right].
\end{align}
We can use these formulae to calculate the cumulant $\langle q_1
\ldots q_n \rangle$, and thus check Theorem \ref{main-theorem}for this special
class of wavefunctions $\phi$. Each formula contains on the right-hand
side a leading sequential weak value, but there are also extra terms,
such as $(A_1)_w({\bar A}_2)_w$ in (\ref{2q}) and $(A_2,A_1)_w({\bar
A}_3)_w$ in (\ref{3q}). All these extra terms are eliminated when the
cumulant is calculated, and we are left with (\ref{main-result}) with
$\xi_{q_1 \ldots q_n}=(1/2)^{n-1}$.

This gratifying simplification depends on the fact that the cumulant
is a sum over all partitions. For instance, it does not occur if one
uses the covariance instead of the cumulant. To see this, look at the
case $n=4$: The term $\langle q_1q_2q_3q_4 \rangle$ in
$Cov(q_1,q_2,q_3,q_4)$, the covariance of pointer positions, gives
rise via (\ref{4q}) to weak value terms like
$(A_4,A_3)_w(\overline{A_2,A_1})_w$. However, (\ref{4covariance})
together with (\ref{2q}), (\ref{3q}) and (\ref{4q}) show that
$Cov(q_1,q_2,q_3,q_4)$ has no other terms that generate any multiple
of $(A_4,A_3)_w(\overline{A_2,A_1})_w$, and consequently this weak
value expression cannot be cancelled and must be present in
$Cov(q_1,q_2,q_3,q_4)$. This means that there cannot be any equation
relating $Cov(q_1,q_2,q_3,q_4)$ and $Cov(A_4,A_3,A_2,A_1)_w$. This
negative conclusion does not apply to the cumulant $\langle q_1q_2q_3q_4
\rangle^c$, as this includes terms such as $\langle q_1q_2
\rangle\langle q_3q_4 \rangle$; see (\ref{4cumulant}).

\section{Simultaneous weak measurement}

We have treated the interactions between each pointer and the system
individually, the Hamiltonian for the $k$'th pointer and system being
$H_k=g_k \delta(t-t_k) s_kA_k$, but of course we can equivalently
describe the interaction between all the pointers and the system by
$H= \sum_k g_k \delta(t-t_k) s_kA_k$. For sequential measurements we
implicitly assume that all the times $t_k$ are distinct. However, the
limiting case where there is no evolution between coupling of the
pointers and all the $t_k$'s are equal is of interest, and is the {\em
simultaneous} weak measurement considered in \cite{RS04,R04,LR05}. In
this case, the state of the pointers after post-selection is given by
\begin{align} \label{simul-bigstate}
\Psi_\cM=\bra{\psi_f}e^{-i(g_1s_1A_1 \ldots +g_ns_nA_n)}
\ket{\psi_i}\phi_1(r_1) \ldots \phi_n(r_n).
\end{align}
The exponential $e^{-i(g_1s_1A_1 \ldots +g_ns_nA_n)}$ here differs
from the sequential expression $e^{-ig_ns_nA_n} \ldots
e^{-ig_1s_1A_1}$ in (\ref{bigstate}) in that each term in the
expansion of the latter appears with the operators in a specific
order, viz. the arrow order $\leftarrow$ as in (\ref{4weak}), whereas
in the expansion of the former the same term is replaced by a
symmetrised sum over all orderings of operators.  For instance, for
arbitrary operators $X$, $Y$ and $Z$, the third degree terms in
$e^Xe^Ye^Z$ include $X^3/3!$, $X^2Y/2!$ and $XYZ$, whose counterparts
in $e^{(X+Y+Z)}$ are, respectively, $X^3/3!$, $\{X^2Y+XYX+YX^2\}/3!$
and $\{XYZ+XZY+YXZ+YZX+ZXY+ZYX\}/3!$.
Apart from this symmetrisation, the calculations in Section
\ref{theorem-section} can be carried through unchanged for
simultaneous measurement. Thus if we replace the sequential weak value
by the {\em simultaneous weak value} \cite{RS04,R04,LR05}
\begin{align} \label{swv}
(A_{i_k}, \ldots, A_{i_1})_{ws}=\frac{1}{k!}\sum_{\pi \in S_k} \left(A_{i_{\pi(k)}}, \ldots, A_{i_{\pi(1)}} \right)_w,
\end{align}
where the sum on the right-hand side includes all possible orders of
applying the operators, we obtain a version of Theorem
\ref{main-theorem} for simultaneous weak measurement:
\begin{align}
\langle r_1 \ldots r_n \rangle^c= g_1 \ldots g_n Re \left\{ \xi
(A_n, \ldots, A_1)^c_{ws}\right\}.
\end{align}
Likewise, relations such (\ref{2q}), (\ref{3q}), etc., hold with
simultaneous weak values in place of the sequential weak values;
indeed, these relations were first proved for simultaneous
measurement \cite{RS04,R04}.

\begin{figure}[hbtp]
\centerline{\epsfig{file=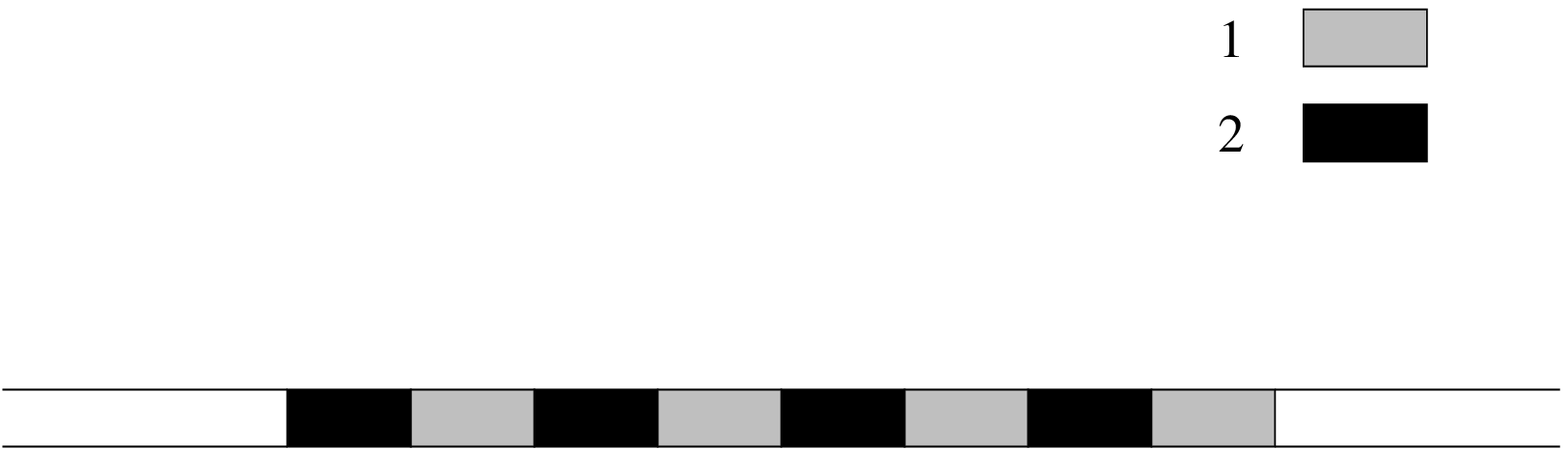,width=0.5\textwidth}}
\caption{\label{alternate}}
\end{figure}

From (\ref{swv}) we see that, when the operators $A_k$ all commute,
the sequential and simultaneous weak values coincide. One important
instance of this arises when the operators $A_k$ are applied to
distinct subsystems, as in the case of the simultaneous weak
measurements of the electron and positron in Hardy's paradox
\cite{Hardy92,ABPRT01}. 

When the operators do not commute, the meaning of simultaneous weak
measurement is not so obvious. One possible physical interpretation
follows from the well-known formula
\begin{align} \label{formula}
e^{X+Y}=\lim_{N \to \infty}(e^{X/N}e^{Y/N})^N
\end{align}
and its analogues for more operators. Suppose two pointers, one for
$A_1$ and one for $A_2$, are coupled alternately in a sequence of $N$
short intervals (Figure \ref{alternate}, top diagram) with coupling
strength $g_k/N$ for each interval. This is an enlarged sense of
sequential weak measurement \cite{MJP07} in which the same pointer is
used repeatedly, coherently preserving its state between
couplings. The state after post-selection is
\begin{align}\label{2simul}
\Psi_\cM=\bra{\psi_f}\left(e^{-i\left(\frac{g_2}{N}\right)s_2A_2}e^{-i\left(\frac{g_1}{N}\right)s_1A_1}\right)^N
\ket{\psi_i}\phi_1(r_1)\phi_2(r_2).
\end{align}
 From (\ref{formula}) we deduce that
\begin{align}
\Psi_\cM \approx \bra{\psi_f}e^{-i(g_2s_2A_2+g_1s_1A_1)}\ket{\psi_i}\phi_1(r_1)\phi_2(r_2).
\end{align}
This picture readily extends to more operators $A_k$.

One can also simulate a simultaneous measurement by averaging the
results of a set of sequential measurements with the operators in all
orders; in effect, one carries out a set of experiments that implement
the averaging in (\ref{swv}). There is then no single act that counts
as simultaneous measurement, but weak measurement in any case relies
on averaging many repeats of experiments in order to extract the
signal from the noise. In a certain sense, therefore, sequential
measurement includes and extends the concept of simultaneous
measurement. However, if we wish to accomplish simultaneous
measurement in a single act, then we need a broader concept of weak
measurement where pointers can be re-used; indeed, we can go further,
and consider generalised weak coupling between one time-evolving
system and another, followed by measurement of the second
system. However, even in this case, the measurement results can be
expressed algebraically in terms of the sequential weak values of the
first system \cite{MJP07}.

\section{Lowering operators}

Lundeen and Resch \cite{LR05} showed that, for a gaussian initial
pointer wavefunction, if one defines an operator $a$ by
\begin{align*} 
a_{LR}=\langle p^2 \rangle_\ci^{1/2} \left(q+\frac{ip}{2\langle p^2 \rangle_\ci}\right),
\end{align*}
then the relationship 
\begin{align*}
\langle a_{LR} \rangle=g \langle p^2 \rangle_\ci^{1/2} A_w
\end{align*}
holds. They argued that $a_{LR}$ can be interpreted physically as a
lowering operator, carrying the pointer from its first excited state
$\ket{1}$, in number state notation, to the gaussian state $\ket{0}$
(despite the fact that the pointer is not actually in a harmonic
potential). Although $a_{LR}$ is not an observable, $\langle a_{LR}
\rangle$ can be regarded as a prescription for combining expecations
of pointer position and momentum to get the weak value.

If instead of $a_{LR}$ one takes
\begin{align} \label{gaussian}
a=q+\frac{ip}{2\langle p^2 \rangle_\ci},
\end{align}
then the even simpler relationship
\begin{align} \label{LR1}
\langle a \rangle=g A_w,
\end{align}
holds. We refer to $a$ as a generalised lowering operator. 

Lundeen and Resch also extended their lowering operator concept to
simultaneous weak measurement of several observables $A_k$. Rephrased
in terms of our generalised lowering operators $a_k$ defined by
(\ref{gaussian}), their finding \cite{LR05} can be stated as
\begin{align}\label{simul}
\langle a_1 \ldots a_n \rangle=g_1 \ldots g_n(A_1 \ldots A_n)_{ws}.
\end{align}
This is of interest for two reasons. First, the entire simultaneous
weak value appears on the right-hand side, not just its real part; and
second, the ``extra terms'' in the simultaneous analogues of
(\ref{2q}), (\ref{3q}) and (\ref{4q}) have disappeared. The lowering
operator seems to relate directly to weak values.

We can generalise these ideas in two ways. First, we extend them from
simultaneous to sequential weak measurements. Secondly, instead of
assuming the initial pointer wavefunction is a gaussian, we allow it
be arbitrary; we do this by defining a generalised lowering operator
\begin{align}\label{lowering}
a=q+i \frac{p}{\eta}, \qquad \mbox{with} \qquad \eta=-i{\overline
\xi}_p/{\overline \xi}_q.
\end{align} 
For a gaussian $\phi$, $\eta=2\langle p^2 \rangle_\ci$, so the above
definition reduces to (\ref{gaussian}) in this case. In general,
however, $\phi$ will not be annihilated by $a$ and is therefore not
the number state $\ket{0}$ (this state is a gaussian with complex
variance $\eta^{-1}$). Nonetheless, there is an analogue of Theorem
\ref{main-theorem} in which the whole sequential weak value, rather
than its real part, appears:
\begin{theorem}[Cumulant theorem for lowering operators]\label{lowering-theorem}
For $n>1$
\begin{align}\label{nlowering}
\langle a_1 \ldots a_n \rangle^c = g_1 \ldots g_n \ \vartheta \  (A_n, \ldots A_1)_w^c,
\end{align}
where $\vartheta$ is given by
\begin{align}\label{constant}
\vartheta=\sum_{(i_1, \ldots i_n) \in \{0,1\}^n} \ \frac{(-1)^{\sum i_j} \xi_{r_{i_1} \ldots r_{i_n}} \left( {\overline \xi}_{r_{1-i_1}} \ldots {\overline \xi}_{r_{1-i_n}} \right) }{2\left({\overline{ \xi}}_{p_1} \ldots {\overline{\xi}}_{p_n} \right)}.
\end{align}
For $n=1$ the
same result holds, but with the extra term $\langle a \rangle_\ci$:
\begin{align}\label{1lowering}
\langle a \rangle = \langle a \rangle_\ci +\vartheta g A_w.
\end{align}
\end{theorem}
\begin{proof}
Put $r_0=q$, $r_1=p$. Then
\begin{align*}
\langle a_1 \ldots a_n \rangle^c &=\langle (q_1+i p_1/\eta_1) \ldots (q_n+i p_n/\eta_n) \rangle^c,\\
&=\sum_{(i_1, \ldots i_n) \in \{0,1\}^n} \ \frac{(-1)^{\sum i_j}\langle r_{i_1} \ldots r_{i_n} \rangle^c \left( {\overline \xi}_{r_{1-i_1}} \ldots {\overline \xi}_{r_{1-i_n}} \right) }{\left({\overline \xi}_{p_1} \ldots {\overline \xi}_{p_n} \right)},\\
&=g_1 \ldots g_n \left[ \vartheta (A_n, \ldots A_1)^c_w + \varpi (\overline{A_n, \ldots A_1})^c_w \right],
\end{align*}
where we used Theorem \ref{main-theorem} to get the last line, and
where $\vartheta$ is given by (\ref{constant}) and $\varpi$ by
\begin{align*}
\varpi=\sum_{(i_1, \ldots i_n) \in \{0,1\}^n} \ \frac{(-1)^{\sum i_j} {\overline \xi}_{r_{i_1} \ldots r_{i_n}} \left( {\overline \xi}_{r_{1-i_1}} \ldots {\overline \xi}_{r_{1-i_n}} \right) }{2\left({\overline \xi}_{p_1} \ldots {\overline \xi}_{p_n} \right)};
\end{align*}
(note the bar over ${\overline \xi}_{r_{i_1} \ldots r_{i_n}}$ that is
absent in the definition of $\vartheta$ by (\ref{constant})). 

We want to prove $\varpi=0$, and to do this it suffices to prove that
the complex conjugate of the numerator is zero, i.e.
\begin{align*}
\varpi^\prime=\sum_{(i_1, \ldots i_n) \in \{0,1\}^n} \ (-1)^{\sum i_j} \xi_{r_{i_1} \ldots r_{i_n}} \left( \xi_{r_{1-i_1}} \ldots \xi_{r_{1-i_n}} \right)=0.
\end{align*}
Let $a_k=\langle q_ks_k \rangle_i$, $b_k=\langle q_k \rangle_i \langle
s_k \rangle_i$, $c_k=\langle p_ks_k \rangle_i$, $d_k=\langle p_k
\rangle_i \langle s_k \rangle_i$. Using the definition of $\xi$ in
(\ref{xi}), the above equation can be written
\begin{align*}
\varpi^\prime/(2^{n+1}(-1)^n)&=\prod_{k=1}^n \left\{ a_k(c_k-d_k)-c_k(a_k-b_k)\right\}
-\prod_{k=1}^n \left\{ b_k(c_k-d_k)-d_k(a_k-b_k)\right\}\\
&=\prod (b_kc_k-a_kd_k)-\prod (b_kc_k-a_kd_k) = 0.
\end{align*}

\end{proof}
Suppose the interaction Hamiltonian has the standard von Neumann form
$H_{int}=gpA$, so $s=p$ in the definition of $\xi$ by equation
(\ref{xi}). Then for $n=1$, since $\overline{ \xi}_p=\xi_p$ and
$\overline{\langle qp \rangle}_\ci=\langle pq \rangle_\ci$,
$\vartheta=(-i)(\xi_q-\overline {\xi}_q)=(-i)(\langle qp
\rangle_\ci-\langle pq \rangle_\ci)=1$, so we get the even simpler
result
\begin{align}\label{simple}
\langle a \rangle=\langle a \rangle_\ci+gA_w.
\end{align}
This is valid for all initial pointer wavefunctions, and therefore
extends Lundeen and Resch's equation (\ref{LR1}). It seems almost too
simple: there is no factor corresponding to $\xi$ in equation
(\ref{qmean}). However, a dependency on the initial pointer
wavefunction is of course built into the definition of $a$ through
$\eta$.

For $n>1$ it is no longer true that $\vartheta=1$, even with the
standard interaction Hamiltonian. However, if in addition $\langle p
\rangle_\ci=0$, then
\begin{align*}
\vartheta=(-i)^n\prod_{k=1}^n\left(\langle q_kp_k \rangle_\ci- \langle p_kq_k \rangle_\ci \right)=(-i)^n(i)^n=1.
\end{align*}
Thus $\langle a_1 \ldots a_n \rangle^c = g_1 \ldots g_n (A_n, \ldots
A_1)_w^c$ for all $n$. Applying the inverse operation for the
cumulant, given by Propostion \ref{anti}, we deduce:
\begin{corollary}
If $\langle p \rangle_\ci=0$, e.g. if the initial pointer wavefunction
$\phi$ is real, then for $n>1$
\begin{align}\label{nice}
\langle a_1 \ldots a_n \rangle=g_1 \ldots g_n (A_n, \ldots A_1)_w.
\end{align}
\end{corollary}
This is the sequential weak value version of the result for
simultaneous measurements, (\ref{simul}), but is more general than the
gaussian case treated in \cite{LR05}.

We might be tempted to try to repeat the above argument for pointer
positions $q_k$ instead of the lowering operators $a_k$ by applying
the anti-cumulant to both sides of (\ref{main-result}). This fails,
however, because of the need to take the real part of the weak values;
in fact, this is one way of seeing where the extra terms come from in
(\ref{2q}), (\ref{3q}) and (\ref{4q}) and their higher analogues.

Note also that (\ref{nice}) does not hold for general $\phi$, since
then different subsets of indices may have different values of
$\vartheta$. 

\section{Discussion}
The procedure for sequential weak measurement involves coupling
pointers at several stages during the evolution of the system,
measuring the position (or some other observable) of each pointer, and
then multiplying the measured values together. In \cite{MJP07} it was
argued that we would really like to measure the product of the values
of the operators $A_1, \ldots A_n$, and that this corresponds to the
sequential weak value $(A_n, \ldots, A_1)_w$. Multiplication of the
values of pointer observables is the best we can do to achieve this
goal. However, this brings along extra terms, such as $(A_1)_w({\bar
A}_2)_w $ in (\ref{2q}), which are an artefact of this method of
extracting information. From this perspective, the cumulant extracts
the information we really want.

In \cite{MJP07}, a somewhat idealised measuring device was being
considered, where the pointer position distribution is real and has
zero mean. When the pointer distribution is allowed to be arbitrary,
the expressions for $\langle q_1 \ldots q_n \rangle$ become wildly
complicated (see for instance (\ref{horrible})). Yet the cumulant of
these terms condenses into the succinct equation (\ref{main-result})
with all the complexity hidden away in the one number $\xi$. Why does
the cumulant have this property?

Recall that the cumulant vanishes when its variables belong to two
independent sets. The product of the pointer positions $q_1, \ldots
q_n$ will include terms that come from products of disjoint subsets of
these pointer positions, and the cumulant of these terms will be sent
to zero, by Lemma \ref{indep-lemma}. For instance, with $n=2$, the
pointers are deflected in proportion to their individual weak values,
according to (\ref{firstxi}), and the cumulant subtracts this
component leaving only the component that arises from the
$O(g^2)$-influence of the weak measurement of $A_1$ on that of
$A_2$. The subtraction of this component corresponds to the
subtraction of the term $(A_1)_w({\bar A}_2)_w $ from (\ref{2q}). In
general, the cumulant of pointer positions singles out the maximal
correlation involving all the $q_i$, and the theorem tells us that
this is directly related to the corresponding ``maximal correlation''
of sequential weak values, $(A_n, \ldots, A_1)^c$, which involves all
the operators.

In fact, the theorem tells us something stronger: that it does not
matter what pointer observable $r(p,q)$ we measure, e.g. position,
momentum, or some Hermitian combination of them, and that likewise the
coupling of the pointer with the system can be via a Hamiltonian
$H_{int}=gs(p,q)A$ with any Hermitian $s(p,q)$. Different choices of
$r$ and $s$ lead only to a different multiplicative constant $\xi$ in
front of $(A_n, \ldots, A_1)_w^c$ in (\ref{main-result}). We always
extract the same function of sequential weak values, $(A_n, \ldots,
A_1)_w^c$, from the system. This argues both for the fundamental
character of sequential weak values and also for the key role played
by their cumulants.

\section{Acknowledgements}

I am indebted to J. {\AA}berg for many discussions and for comments on
drafts of this paper; I thank him particularly for putting me on the
track of cumulants. I also thank A. Botero, P. Davies, R. Jozsa,
R. Koenig and S. Popescu for helpful comments. A preliminary version
of this work was presented at a workshop on ``Weak Values and Weak
Measurement'' at Arizona State University in June 2007, under the
aegis of the Center for Fundamental Concepts in Science, directed by
P. Davies.

\appendix

\section{An explicit calculation}

To calculate $\langle q_1q_2 \rangle$ for arbitrary pointer
wavefunctions $\phi_1$ and $\phi_2$, we use (\ref{bigstate}) to
determine the state of the two pointers after the weak interaction,
and then evaluate the expectation using (\ref{expectation}), keeping
only terms up to order $g^2$. We define
\begin{align*}
\mu_k&=\langle q_k \rangle_\ci,\ \
\nu_k=\langle p_k \rangle_\ci,\ \
\zeta_k=\langle p_k^2 \rangle_\ci,\\
\rho_k&=\langle q_kp_k \rangle_\ci,\ \
\sigma_k=\langle q_kp_k^2 \rangle_\ci,\ \
\tau_k=\langle p_kq_kp_k \rangle_\ci,
\end{align*}

Then, expanding the exponential in (\ref{bigstate}) and substituting
$\Psi$ in (\ref{expectation}) gives, up to order $g^2$,
\begin{align}\label{horrible}
\langle q_1q_2 \rangle&=\mu_1\mu_2-ig_1\left\{ \left( (A_1)_w-({\bar
A}_1)_w \right)\mu_1\nu_1\mu_2- ({\bar A}_1)_w {\bar
\rho}_1\mu_2+(A_1)_w\rho_1\mu_2 \right\}\\
\nonumber &-ig_2\left\{\left((A_2)_w-({\bar
A}_2)_w\right)\mu_1\mu_2\nu_2 -({\bar A}_2)_w \mu_1{\bar
\rho}_2+(A_2)_w\mu_1\rho_2 \right\}\\
\nonumber &+g_1^2 \left\{|(A_1)_w|^2(\tau_1\mu_2-\mu_1\zeta_1\mu_2)
+\left( (A_1^2)_w+ ({\bar A}_1^2)_w \right)\frac{\mu_1\zeta_1\mu_2}{2}
\right\}\\
\nonumber &-g_1^2\left\{\left( (A_1)_w-({\bar A}_1)_w
\right)^2\mu_1\nu_1^2\mu_2 +(A_1^2)_w\frac{\sigma_1\mu_2}{2}+({\bar
A}_1^2)_w\frac{{\bar \sigma_1}\mu_2}{2} \right\}\\
\nonumber &+g_2^2 \left\{ |(A_2)_w|^2(\mu_1\tau_2-\mu_1\mu_2\zeta_2)
+\left( (A_2^2)_w+ ({\bar A}_2^2)_w
\right)\frac{\mu_1\mu_2\zeta_2}{2}\right \}\\ 
\nonumber &-g_2^2 \left\{\left( (A_2)_w-({\bar A}_w)_2
\right)^2\mu_1\mu_2\nu_2^2+(A_2^2)_w \frac{\mu_1\sigma_2}{2}+({\bar
A}_2^2)_w\frac{\mu_1{\bar \sigma_2}}{2}\right \}\\
\nonumber &+g_1g_2 \left\{(A_1)_w({\bar A}_2)_w \rho_1 {\bar
\rho_2}+({\bar A}_1)_w(A_2)_w {\bar \rho_1}\rho_2
-(A_2,A_1)_w\rho_1\rho_2-\overline{(A_2,A_1)}_w{ \bar\rho}_1{\bar
\rho}_2\right \}\\
\nonumber &-g_1g_2 \left\{ 2\left( (A_1)_w-({\bar A}_1)_w
 \right)\left( (A_2)_w-({\bar A}_2)_w \right)\mu_1\nu_1\mu_2\nu_2
 \right\}\\
\nonumber &+g_1g_2 \left\{ \left( (A_2,A_1)_w+\overline{(A_2,A_1)}_w
-(A_1)_w({\bar A}_2)_w-({\bar A}_1)_w(A_2)_w
\right)\mu_1\nu_1\mu_2\nu_2 \right \}\\ 
\nonumber &+g_1^2 \left\{ \left( (A_1)_w-({\bar A}_1)_w\right)(A_1)_w
\nu_1\rho_1\mu_2-\left( (A_1)_w-({\bar A}_1)_w\right) ({\bar
A_1})_w\nu_1{\bar \rho}_1\mu_2 \right \} \\
\nonumber &+g_2^2 \left\{ \left( (A_2)_w-({\bar A}_2)_w\right)(A_2)_w
\mu_1\nu_2\rho_2-\left( (A_2)_w-({\bar A}_2)_w\right)({\bar A}_2)_w
\mu_1\nu_2{\bar \rho_2} \right \} \\
\nonumber &+g_1g_2 \left\{ \left( (A_1)_w-({\bar
A}_1)_w\right)(A_2)_w\mu_1\nu_1\rho_2-\left( (A_1)_w-({\bar
A}_1)_w\right)({\bar A}_2)_w\mu_1\nu_1{\bar \rho_2} \right \}\\
\nonumber &+g_1g_2 \left\{ \left( (A_2)_w-({\bar
A}_2)_w\right)(A_1)_w\rho_1\mu_2\nu_2-\left( (A_2)_w-({\bar
A}_2)_w\right)({\bar A}_1)_w{\bar \rho_1}\mu_2\nu_2 \right \}.
\end{align}

To calculate the cumulant $\langle q_1,q_2 \rangle^c=\langle q_1q_2
\rangle-\langle q_1 \rangle\langle q_2 \rangle$ we need $\langle q
\rangle$ up to order $g^2$:
\begin{align}\label{singleq}
\langle q \rangle&=\mu+ig\left\{A_w(\mu\nu-\rho)-{\bar
A}_w(\mu\nu-{\bar \rho})\right\}+g^2
|A_w|^2\left(\tau-\mu\zeta+2\mu\nu^2-\nu\rho-\nu{\bar \rho}\right)\\
\nonumber &+g^2\left\{ (A^2)_w\left(\frac{\mu\zeta}{2}-\frac{\sigma}{2}\right)-({\bar
A}^2)_w\left(\frac{\mu\zeta}{2}-\frac{\bar \sigma}{2}\right)
+(A_w)^2(\nu\rho-\mu\nu^2)+({\bar A}_w)^2(\nu{\bar \rho}-\mu\nu^2)
\right \}.
\end{align}

Substituting from (\ref{horrible}) and (\ref{singleq}) a radical
simplification occurs:
\begin{align}\label{explicit}
\langle q_1q_2 \rangle^c=g_1g_2 \left\{ (A_2,A_1)_w-(A_1)_w(A_2)_w \right\}\left(\mu_1\nu_1\mu_2\nu_2-\rho_1\rho_2\right) + \mbox{complex conjugate}.
\end{align}
This, of course, is what Theorem \ref{main-theorem} tells us.

\end{document}